\documentclass{article}
\usepackage{graphicx} 
\usepackage[utf8]{inputenc}
\usepackage{braket}
\usepackage{amsfonts}
\usepackage{amsmath}
\usepackage{amssymb}
\usepackage{amsthm}
\usepackage{mathrsfs}
\usepackage[a4paper,total={6in,9in}]{geometry}
\usepackage{setspace}
\usepackage{verbatim}
\usepackage{dsfont}
\usepackage{tikz}
\usepackage{tikz-cd}
\usepackage[dvipsnames]{xcolor}
\usetikzlibrary{arrows.meta}
\usetikzlibrary{calc}
\usetikzlibrary{patterns}
\usepackage{comment}
\usepackage{centernot}
\usepackage{arydshln}
\usepackage{mathtools}
\usepackage{afterpage}

\usepackage{appendix}
\usepackage{url}
\usepackage{bm}

\usepackage[style=numeric, sorting=none]{biblatex} 
\newtheorem{theorem}{Theorem}[section]
\newtheorem{proposition}[theorem]{Proposition}
\newtheorem{lemma}[theorem]{Lemma}
\newtheorem{corollary}[theorem]{Corollary}
\newtheorem{definition}[theorem]{Definition}
\newtheorem{conjecture}[theorem]{Conjecture}

\newtheorem{example}[theorem]{Example}

\newcommand{\Aut}{\mathrm{Aut}}
\newcommand{\Iso}{\mathrm{Iso}}
\newcommand{\diam}{\mathrm{diam}}

\newcommand{\loc}{\mathrm{loc}}
\newcommand{\bnd}{\mathrm{bnd}}

\newcommand{\phys}{\mathrm{phys}}

\newcommand\blfootnote[1]{%
  \begingroup
  \renewcommand\thefootnote{}\footnote{#1}%
  \addtocounter{footnote}{-1}%
  \endgroup
}

\title{Pathology-Free Real-Space Renormalization Group Theory on an Inverse Limit Space}
\bigskip
\author{Fabio Arz}

\date{Albert Einstein Center for Fundamental Physics\\
        Institute for Theoretical Physics\\
        University of Bern, Sidlerstrasse 5, CH-3012 Bern, Switzerland\\[2ex]
        \today}

\begin{document}

\maketitle

\begin{abstract}
    It has been over fifty years since Kenneth Wilson had his Nobel-prize-winning ideas on the renormalization group. In this time frame, despite many attempts, no mathematical results have implemented Wilson's vision to a satisfactory degree. Although a number of predictions stemming from the renormalization group framework have been proven to date, these proofs usually rely on alternative ideas and do not cover the full predictive power of Wilson's renormalization group. The discovery of pathologies within rigorous implementations of renormalization group transformations have further reduced the enthusiasm of the mathematics community for this subject. This paper presents a novel framework to implement Wilson's original idea in a rigorous manner based on inverse limits of finite systems. Doing so, one avoids these pathologies and hence obtains a mathematically sound framework for Wilson's renormalization group. This opens up a completely new path to tackle some of the existing problems in rigorous renormalization group theory like the existence of non-trivial fixed points.
\end{abstract}

\blfootnote{*E-mail: fabio.arz@unibe.ch}

\newpage

\section{Introduction}

The origin of the renormalization group (RG) stems from the late 1940s when physicists proposed renormalization techniques in order to obtain control over divergences in the nascent field of quantum field theory. Fast forward roughly twenty years to 1966 when Leo Kadanoff \cite{Kadanoff} introduced the heuristics of block-spin transformations to investigate the $2\mathrm{D}$ Ising model close to the critical point. The method behind these transformations is to partition the lattice into blocks and then average the spins per block into a single block-spin. In the beginning of the 70s Kenneth Wilson \cite{WilsonI},\cite{WilsonII} combined these block-spin transformations with the concept of RG stemming from quantum field theory into a complete framework with applications in both classical statistical mechanics as well as quantum field theory. His vision was that consecutive applications of so called renormalization group transformations (RGTs) induce a flow in the space of couplings, mapping models to their counterparts viewed on coarser length scales. While these transformations may change the microscopic details by a lot, they leave the macroscopic behaviour essentially untouched. Wilson then postulated that these flows converge to scale (and even conformal) invariant fixed points which determine the macroscopic behaviour of all models in their basin of attraction. This then divides the possible macroscopic properties into a set of universality classes each depending on its own fixed point. Among many other things, this framework allowed for an explanation of previously obtained experimental data that wildly different systems behave very similarly at criticality. Although, from a heuristic point of view, Wilson's RG-theory was (and has been so to this day) extremely successful, it possessed a number of issues preventing it from becoming a fully rigorous concept.\smallskip 

Early in the development of a rigorous theory for RGTs on lattice models, several mathematical physicists noticed certain inconsistencies in their investigations (see for example Griffiths and Pearce \cite{Griffiths1979} as well as Israel \cite{Israel1981}). In a landmark paper for rigorous RG, van Enter, Fernández and Sokal \cite{van_Enter_1993} summarized and extended these findings and gave a proper explanation on the origin of the pathological behaviour of such transformations. In order to understand these pathologies, we first have to find a common ground for the definition of real-space RGTs. As there are dozens of introductory texts on the theory of lattice models, this part is kept rather short. A very detailed overview can be found in the text book by Friedli and Velenik \cite{friedli_velenik_2017}. For now, let us concentrate our attention on the standard hypercubic lattice with vertex set $\mathds{Z}^d$. In case of the Ising model, we shall then consider the single spin-space $\Omega_0:=\{-1,1\}$. The physical configuration space is obtained by assigning a spin to each vertex, i.e.\ $\Omega:=\Omega_0^{\mathds{Z}^d}$. The last ingredient we require is the $\sigma$-algebra\footnote{In case the reader is not very familiar with measure theory, one can ignore these things as they will not be relevant beyond this short introduction.}
\begin{equation}
    \mathscr{F}:=\sigma(\{\mathcal{A}\subset 2^\Omega\,:\,\mathcal{A}\,\,\mathrm{is\,\,finitely\,\,generated}\}).
\end{equation}
Finitely generated means that there exists some finite subgraph $\Lambda\subset\mathds{Z}^d,|\Lambda|<\infty$ together with some\footnote{$\Omega_\Lambda$ refers to the space $\{-1,1\}^\Lambda$.} $A\subset2^{\Omega_\Lambda}$ such that
\begin{equation}
    \sigma\in \mathcal{A}\iff\sigma|_\Lambda\in A.
\end{equation}
Let us then denote $\mathscr{M}_1(\Omega)$ the space of all probability measures on the measurable space $(\Omega,\mathscr{F})$. A real-space RGT comes in the form of a Markov kernel
\begin{equation}
    T:\mathscr{F}\times\Omega\to[0,1],
\end{equation}
which induces a map 
\begin{equation}
    T:\mathscr{M}_1(\Omega)\to\mathscr{M}_1(\Omega),\quad\mu\mapsto T\mu,\qquad T\mu[\mathcal{A}]:=\int T(\mathcal{A},x)\mu(\mathrm{d}x).
\end{equation}
In the case that $T$ is constructed from a coarse-graining procedure on the lattice (a more precise definition follows below), we refer to the induced action on $\mathscr{M}_1(\Omega)$ as a real-space RGT. From a physical point of view, not all measures in $\mathscr{M}_1(\Omega)$ are equally interesting and one mostly works with Gibbs measures. These are a family of probability measures that arise from well-behaved interactions (in Section \ref{Sec2}, we will give a formal definition for absolutely summable and translation invariant interactions which are sufficiently regular for the theory of Gibbs measures). Historically, one would usually start in a finite volume $\Lambda\subset\mathds{Z}^d,|\Lambda|<\infty$ and then parametrize the measure via $\mu\propto e^{-H}$ where $H:\Omega_\Lambda\to\mathds{R}$ is the corresponding Hamilton function describing the energy of a given configuration. Afterwards, one would consider weak convergence of measures when $\Lambda$ approaches $\mathds{Z}^d$, which is known as the thermodynamic limit. In the late 60s Dobrushin \cite{doi:10.1137/1113026} as well as Lanford and Ruelle \cite{Lanford1969} independently discovered the idea behind the DLR equations. These allow for a consistency relation between interactions and probability measures directly in the infinite volume space $\Omega$ rather than working with finite volumes and subsequent weak convergence instead. Let us for now denote with $\mathscr{V}$ the space of interactions equipped with a suitable norm (see \cite{van_Enter_1993} for a discussion on the possible choices of norms and their physical relevance). Wilson's original idea is that the RG induces a flow in the space of couplings (interactions). Mathematically, this requires a map $R_T:\mathscr{V}\to\mathscr{V}$ such that the following diagram commutes:
\[
\begin{tikzcd}
\mathscr{V} \arrow[r, "R_T"] \arrow[d, "DLR"] & \mathscr{V}  \\
\mathscr{M}_1(\Omega) \arrow[r, "T"]                 & \mathscr{M}_1(\Omega) \arrow[u, "DLR"]
\end{tikzcd}
\]
However, in \cite{van_Enter_1993} the authors have rigorously proved for many different physically relevant choices of $T$ that the upward arrow and with it the map $R_T$ in the above diagram fail to exist for certain input interactions, famously also for the classical nearest-neighbour interaction at sufficiently low temperatures. This means that the RGT maps some measures consistent with an interaction to probability measures that are no longer consistent with any interaction. While $T$ is still globally defined on $\mathscr{M}_1(\Omega)$, which allows us to investigate some properties of the RG-flow like the fixed-point structure of $T$, a number of predictions from Wilson's original idea require the RG to act on interactions instead. Although this is definitely a set back for rigorous implementations of Wilson's RG, it is not a complete failure. Most predictions of RG-theory concern themselves with the physics at second order phase transitions. So far, the proven pathologies primarily apply to low temperatures and only in very few cases affect temperatures up to or even beyond the critical temperature. Thus, so far these pathologies do not rule out that in the interesting region, RGTs exhibit the expected behaviour. It remains an open question to determine the maximal support of $R_T$ within $\mathscr{V}$ for any given $T$ and whether such a map is well-defined for models at a second order phase transition. Beyond these issues of well-definedness there are further mathematical issues arising when one attempts to rigorously implement RGTs. In particular, in \cite{Yin} Yin has proven that even in the region where the RG is well-defined, the linearization of the RGT around the trivial $T=\infty$ (non-interacting) fixed point is not as regular as one would have hoped. This has huge ramifications for RG-theory as the linearization of the RGT around the fixed point is supposed to contain data on the macroscopic behaviour of the universality class, in particular the critical exponents. All in all, the mathematical foundation of Wilson's RG for classical lattice models is still lacking in many areas more than 50 years later.\medskip

This paper proposes a new approach to rigorous implementations of the RG that avoids the previously mentioned pathologies completely. The idea is to consider inverse limit points of systems defined on increasing finite volumes. On finite lattices, there is a one to one correspondence between Hamilton functions (modulo a constant) and probability measures that are non-zero everywhere. Therefore, a single application of a given RGT on a finite lattice can always be extended to the space of Hamilton functions. However, on finite lattices the RGT usually maps one space to another. The difficulty then lies in constructing a space that is mapped into itself. It is hoped that within this approach, also the linearization of the RGT around fixed points becomes well-behaved. It is worth to mention that there are other existing ideas to obtain a rigorous implementation of the RG, one of them being the tensor network formulation of RGTs. See for example the paper by Kennedy and Rychkov \cite{Slava}. \medskip

This article is structured as follows. The remainder of the introduction is spent to explain all the necessary concepts regarding the Ising model on finite graphs. Section \ref{Sec2} provides the definition of block-spin RGTs and connects them to the ideas of inverse limits. Afterwards, we will prove the main theorem which states that there is a suitable subset in the space of inverse limits of interactions on which the RGT is well-defined and thus allows for a pathology free implementation of Wilson's RG. Section \ref{Sec3} will picture the current standpoint and then highlight the next logical steps towards realizing the complete Wilsonian picture of RG.

\subsection{Basic Graph-Theory and Ising Spin Models on Finite Graphs}

\begin{definition}
    A graph $G$ is a triplet $(V,E,\phi)$ where $V=\{v_i, i\in I\}$ is a set of vertices and $E=\{e_j, j\in J\}$ is a set of edges, with $I,J$ being at most countable index sets. The incidence function $\phi$ is a map
    \begin{equation}
        \phi:E\to\left\{\{x,y\}\,|\,x,y\in V\right\}.
    \end{equation}
\end{definition}

The following list contains some basic notions from graph theory.
\begin{enumerate}
    \item[i)] We call two vertices $x,y\in V$ adjacent (denoted as $x\sim y$) if $\phi^{-1}(\{x,y\})\neq\emptyset$.
    \item[ii)] We call a graph connected if for any pair of vertices $x,y\in V$ there exists a finite sequence $\{x=v_1,v_2,...,v_N=y\}\subset V$ satisfying $v_i\sim v_{i+1}$. Such a finite sequence is called a path between $x$ and $y$. In case there exists a path between $x,y$ we call them connected.
    \item[iii)] A graph homomorphism $f:G\to H$ is a pair of maps $f = (f_V, f_E)$ where $f_V:V(G)\to V(H)$ and $f_E:E(G)\to E(H)$ such that $\phi_G(e)=\{x,y\}\implies\phi_H(f_E(e))=\{f_V(x),f_V(y)\}$.
    
    If both $f_V$ and $f_E$ are bijective, we call $f$ an isomorphism. The set of all isomorphisms between $G$ and $H$ is denoted as $\Iso(G,H)$, and we write $\Iso(G):=\Iso(G,G)$ for the automorphism group of $G$.
    \item[iv)] A subgraph $\Gamma\subset G$ is a graph $(V(\Gamma), E(\Gamma), \phi_\Gamma)$ where $V(\Gamma)\subset V(G)$, $E(\Gamma)\subset E(G)$, and $\phi_\Gamma = \phi_G\big|_{E(\Gamma)}$.
    \item[v)] For $x\in V(G)$, we denote the set of incident edges at $x$ as $E_G(x):=\{e\in E(G)\,:\,x \in \phi_G(e)\}$.
    \item[vi)] The vertex boundary of a subgraph $\Gamma\subset G$ denoted as $\partial\Gamma$ is given by $\{x\in V(\Gamma)\,:\,\exists y\in V(G)\backslash V(\Gamma),x\sim y\}$.
    \item[vii)] The graph distance between two vertices $x,y\in G$ is given by
    \begin{equation}
        d(x,y)=\begin{cases}\min\{n\,:\,\exists\,\,\mathrm{a\,\,path\,\,of\,\,length}\,\,n\,\,\mathrm{between}\,\,x\,\,\mathrm{and}\,\,y\}&\mathrm{if}\,\,x,y\,\,\mathrm{are\,\,connected}\\
        \infty&\mathrm{otherwise}
        \end{cases}.
    \end{equation}
    For a subset $S\subset V(G)$ we set $d(x,S):=\min_{y\in S}d(x,y)$. The diameter of a finite subgraph $\Gamma\subset G$ is defined as $\diam(\Gamma):=\max_{x,y\in\Gamma}d(x,y)$ where the distance is measured in $G$.
\end{enumerate}

\begin{definition}
    A graph covering is a surjective homomorphism $\pi = (\pi_V, \pi_E): G \to B$ such that for every vertex $x \in V(G)$, the restriction of the edge map
    \begin{equation}
        \pi_E\big|_{E_G(x)}: E_G(x) \to E_B(\pi_V(x))
    \end{equation}
    is a bijection. Therefore, locally $G$ and $B$ look alike. Further, a deck transformation is an $h \in \Iso(G)$ satisfying $\pi \circ h = \pi$. We denote the space of all deck transformations as $\mathcal{D}(\pi) \subset \Iso(G)$.
\end{definition}

\begin{definition}
    Let $B$ be a connected graph and $\pi:G\to B$ a graph covering. A fundamental domain of this covering is a connected subgraph $Q\subseteq G$ such that $\pi_V|_{V(Q)}:V(Q)\to V(B)$ is a bijection.
\end{definition}

\begin{example}\label{Ex1}
    We shall consider graphs $G,B$ where $V(G)=E(G)=\mathds{Z}$ and $\phi_G(n)=\{n,n+1\}$. The base graph $B$ is given by $V(B)=E(B)=\mathds{Z}/4\mathds{Z}$ and\,\footnote{$[n]$ denotes the equivalence class $\{ m\in\mathds{Z}\,:\,m-n\mod4=0\}$.}
    \begin{equation}
        \phi_B([n])=\{[n],[n+1]\}.
    \end{equation}
    We then find the following covering map
    \begin{equation}
        \pi_V:V(G)\to V(B),\quad n\mapsto [n]\quad;\quad\pi_E:E(G)\to E(B),\quad n\mapsto [n]
    \end{equation}
    which can be visualized in the following way:
    
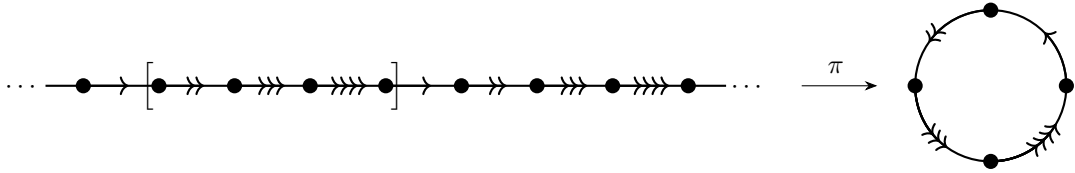
\begin{center}
        \begin{tikzpicture}
            \draw[thick][black][-] (-5.5,0) to (-5,0);
            \draw[thick][black][-] (-5,0) to (-4,0);
            \draw[thick][black][->] (-5,0) to (-4.4,0);
            \draw[thick][black][-] (-4,0) to (-3,0);
            \draw[thick][black][->>] (-4,0) to (-3.4,0);
            \draw[thick][black][-] (-3,0) to (-2,0);
            \draw[thick][black][->>>] (-3,0) to (-2.3333,0);
            \draw[thick][black][-] (-2,0) to (-1,0);
            \draw[thick][black][->>>>] (-2,0) to (-1.25,0);
            \draw[thick][black][-] (-1,0) to (0,0);
            \draw[thick][black][->] (-1,0) to (-0.4,0);
            \draw[thick][black][-] (0,0) to (1,0);
            \draw[thick][black][->>] (0,0) to (0.6,0);
            \draw[thick][black][-] (1,0) to (2,0);
            \draw[thick][black][->>>] (1,0) to (1.6666,0);
            \draw[thick][black][-] (2,0) to (3,0);
            \draw[thick][black][->>>>] (2,0) to (2.75,0);
            \draw[thick][black][-] (3,0) to (3.5,0);
            
            \draw[thick][black] (8,0) arc[start angle=0, end angle=90, radius=1.0];
            \draw[thick][black][->] (8,0) arc[start angle=0, end angle=45, radius=1.0];
            \draw[thick][black] (7,1) arc[start angle=90, end angle=180, radius=1.0];
            \draw[thick][black][->>] (7,1) arc[start angle=90, end angle=145, radius=1.0];
            \draw[thick][black] (6,0) arc[start angle=180, end angle=270, radius=1.0];
            \draw[thick][black][->] (6,0) arc[start angle=180, end angle=235, radius=1.0];
            \draw[thick][black][->] (6,0) arc[start angle=180, end angle=228, radius=1.0];
            \draw[thick][black][->] (6,0) arc[start angle=180, end angle=221, radius=1.0];
            \draw[thick][black] (7,-1) arc[start angle=270, end angle=360, radius=1.0];
            \draw[thick][black][->] (7,-1) arc[start angle=270, end angle=330, radius=1.0];
            \draw[thick][black][->] (7,-1) arc[start angle=270, end angle=323, radius=1.0];
            \draw[thick][black][->] (7,-1) arc[start angle=270, end angle=316, radius=1.0];
            \draw[thick][black][->] (7,-1) arc[start angle=270, end angle=309, radius=1.0];

            \fill[black] (-5,0) circle (0.1);
            \fill[black] (-4,0) circle (0.1);
            \fill[black] (-3,0) circle (0.1);
            \fill[black] (-2,0) circle (0.1);
            \fill[black] (-1,0) circle (0.1);
            \fill[black] (0,0) circle (0.1);
            \fill[black] (1,0) circle (0.1);
            \fill[black] (2,0) circle (0.1);
            \fill[black] (3,0) circle (0.1);

            \fill[black] (6,0) circle (0.1);
            \fill[black] (7,-1) circle (0.1);
            \fill[black] (8,0) circle (0.1);
            \fill[black] (7,1) circle (0.1);

            \node[] at (-5.8,-0.02) {$\cdots$};
            \node[] at (3.8,-0.02) {$\cdots$};
            \node[] at (4.95,0.25) {$\pi$};

            \node[] at (-4.15,0) {$\Big[$};
            \node[] at (-0.85,0) {$\Big]$};

            \draw[-Stealth] (4.5,0) -- (5.5,0);
            
        \end{tikzpicture}
    \end{center}
    The brackets highlight one possible fundamental domain $Q$. The set of deck transformations for this covering is given by
    \begin{equation}
        \mathcal{D}(\pi)=\{h\in\Iso(G)\,:\,h_V(n)=n+k,h_E(n)=n+k,\quad k\in 4\mathds{Z}\}.
    \end{equation}
\end{example}

\begin{proposition}
    Let $G,B$ be finite graphs, $B$ connected and $\pi:G\to B$ a graph covering. Then there exists some constant $b\geq1$ satisfying
    \begin{equation}
        |\pi_V^{-1}(x)|=b\quad\forall x\in V(B).
    \end{equation}
    We refer to $b$ as the degree of $\pi$.
\end{proposition}

\begin{proof}
    In case that $|V(B)|=1$ the statement is trivial. For $|V(B)|>1$ we consider $x,y\in V(B)$ such that $x\sim y$. Let us denote the set $\mathcal{E}:=\{e\in E_G\,:\,\pi_E(e)\in\phi_B^{-1}(\{x,y\})\}$. We then find
    \begin{equation}
        |\mathcal{E}|=|\pi_V^{-1}(x)|\cdot|\phi_B^{-1}(\{x,y\})|=|\pi_V^{-1}(y)|\cdot|\phi_B^{-1}(\{x,y\})|.
    \end{equation}
    Since $x\sim y\implies|\phi_B^{-1}(\{x,y\})|>0$ we conclude $|\pi_V^{-1}(x)|=|\pi_V^{-1}(y)|$ for adjacent vertices. In case that $B$ is connected, we may use that for arbitrary vertices $x,y\in V(B)$ there is a path of adjacent vertices $\{x=x_0,...,x_N=y\}$ such that
    \begin{equation}
        |\pi_V^{-1}(x)|=|\pi_V^{-1}(x_1)|=\cdots=|\pi_V^{-1}(x_{N-1})|=|\pi_V^{-1}(y)|.
    \end{equation}
\end{proof}

Let $G,B$ be finite graphs, $B$ connected and $\pi:G\to B$ a graph covering of degree $b$. Then there exist fundamental domains $Q_1,...,Q_b\subset G$ such that
\begin{enumerate}
    \item[i)] $V(G)=\bigcup_{i=1}^bV(Q_i)$.
    \item[ii)] $V(Q_i)\cap V(Q_j)=\emptyset,i\ne j$. 
\end{enumerate}
In that case, we call the $Q_1,...,Q_b$ a partition of $G$.

\subsubsection{Ising Spin Models on Finite Graphs}

Let $G$ be some finite graph. An Ising spin is an element of the space $\Omega_0:=\{-1,1\}$. An Ising spin model on $G$ is a probability space $(\Omega_G,2^{\Omega_G},\mu)$ where the configuration space is given by
\begin{equation}
    \Omega_G=\{-1,1\}^{V(G)}.
\end{equation}
In order to restrict ourselves to physically reasonable probability measures, we shall first introduce the concept of symmetries.

\begin{definition}
    A symmetry is a symmetry group $S$ together with an action
    \begin{equation}
        \rho:S\to\Aut(\Omega_G).
    \end{equation}
\end{definition}
A symmetry divides the configuration space into equivalence classes
\begin{equation}
    [\sigma]=\{\tau\in\Omega_G\,|\,\exists s\in S\,\,\mathrm{s.t.}\,\,s\cdot\tau=\sigma\}.
\end{equation}
The space of all equivalence classes is usually denoted as $\Omega_G /S$. There are two distinct types of symmetries. There are internal symmetries which act on the single-spin space $\Omega_0$. Their action on the configurations is then given by
\begin{equation}
    (s\cdot\sigma)_x=s\cdot\sigma_x\,.
\end{equation}
For Ising spins, the only possible internal symmetry is a $\mathds{Z}_2$-action given by\footnote{Here we shall think of $\mathds{Z}_2$ as the set $\{-1,1\}$ with multiplication as the group operation.}
\begin{equation}
    \pm1\cdot\sigma_x=\pm\sigma_x\,.
\end{equation}
The second class of symmetries are spatial symmetries. Their symmetry groups are subgroups $S\subset\Iso(G)$ and act on the configuration via
\begin{equation}
    (s\cdot\sigma)_x=\sigma_{s^{-1}\cdot x}\,.
\end{equation}
In case we consider a covering $\pi:G_\infty\to G$ we may choose the symmetry group $\Iso(G_\infty)\subset\Iso(G)$. The action onto $G$ is then defined via the consistency relation
\begin{equation}
    s\circ\pi=\pi\circ s.
\end{equation}
One might notice that $\Iso(G_\infty)$ is not actually a subgroup of $\Iso(G)$ as it contains too many elements. i.e.\ there are non-trivial elements $s\in\Iso(G_\infty)$ which, however, act trivially on $G$. Therefore, we must consider the coset
\begin{equation}
    \Iso(G_\infty)/\mathcal{D}(\pi)\subset\Iso(G).
\end{equation}
Once we have a chosen a symmetry group $S$, we want to specify a model by selecting a probability measure that is invariant under the symmetry action. Since the configuration space $\Omega_G$ is finite, things become a lot simpler. A probability measure $\mu$ on $(\Omega_G,2^{\Omega_G})$ is uniquely defined by the probabilities of the singletons $\{\sigma\},\sigma\in\Omega_G$. Therefore, we may identify a probability measure on $(\Omega_G,2^{\Omega_G})$ with a function\footnote{In a more formal language, we would refer to $\mu(\sigma)$ as the probability mass function regarding the counting measure.}
\begin{equation}
    \mu:\Omega_\Lambda\to[0,1]
\end{equation}
satisfying
\begin{equation}
    \sum_{\sigma\in\Omega_\Lambda}\mu(\sigma)=1.
\end{equation}
An $S$-invariant probability measure is then identified by such a function satisfying
\begin{equation}
    \mu\circ s^{-1}=\mu\quad\forall s\in S.
\end{equation}
We will denote the space of all $S$-invariant probability measures as $\mathscr{M}_1(\Omega_G/S)$. In statistical mechanics, most of the time one does not declare the probability measure directly but instead the underlying Hamilton function $H:\Omega_G\to\mathds{R}$. The Hamilton function then determines a probability measure via the corresponding Boltzmann weights (the inverse temperature $\beta$ is absorbed into the definition of $H$)
\begin{equation}
    \mu_H(\sigma):=\frac{1}{Z}e^{-H(\sigma)},\quad Z=\sum_{\tau\in\Omega_G}e^{-H(\tau)}.
\end{equation}
Therefore, we are also interested in the space of all $S$-invariant real-valued functions on $\Omega_G$ denoted as $\mathcal{F}(\Omega_G/S)$. The Boltzmann weights relate points in $\mathcal{F}(\Omega_G/S)$ with points in $\mathscr{M}_1(\Omega_G/S)$, although there are two important caveats. From the definition above it is clear that adding a constant term to the Hamilton function leaves the corresponding measure invariant, i.e.\ $\mu_H=\mu_{H+c}$ for some constant $c$. On the other hand, we note that $e^{-H(\sigma)}$ cannot become zero such that we cannot parametrize measures that satisfy $\mu(\sigma)=0$ for some configurations. These measures make up the topological boundary of $\mathscr{M}_1(\Omega_G/S)$. Thus, the correct statement reads
\begin{equation}
    \mathcal{F}(\Omega_G/S)/\mathds{R}\cong\mathscr{M}_1(\Omega_G/S)\backslash\partial\mathscr{M}_1(\Omega_G/S),
\end{equation}
where $\mathcal{F}(\Omega_G/S)/\mathds{R}$ denotes the space of functions modulo a constant and $\cong$ refers to diffeomorphic in this context.

\section{RGTs on the Space of Inverse Limits}\label{Sec2}

Before we concern ourselves with the challenges of implementing RGTs on finite graphs, let us consider the more traditional picture. Therefore, we quickly want to review the example from the introduction where we considered probability measures on the space $(\Omega,\mathscr{F})$ where $\Omega=\{-1,1\}^{\mathds{Z}^d}$ and $\mathscr{F}$ is the $\sigma$-algebra generated by finite events. In order to write down a block-spin transformation for such a system, we partition $\mathds{Z}^d$ into $d$-dimensional unit cubes. To each such cube we then assign a new vertex $x'$ in the centre such that these new vertices form a 'coarser' version of the initial lattice $\mathds{Z}^d$. 

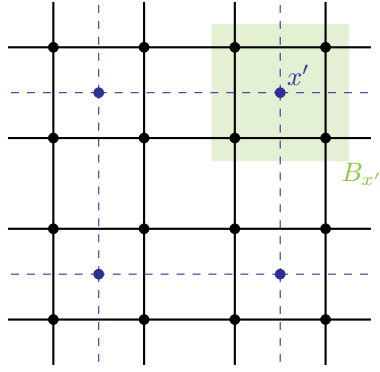
\begin{figure}[h!]
    \centering
    \centering
    \begin{tikzpicture}[scale= 0.6]
        \def \rows {4}  
        \def \cols {4}  
        \def \spacing {2}  
        \def \extend {0.0}  

        \filldraw[LimeGreen!20] ({8.5}, {8.5}) -- ({8.5}, {5.5}) -- ({5.5}, {5.5}) -- ({5.5}, {8.5}) -- cycle;
    
        \foreach \x in {1,...,\cols} {
            \foreach \y in {1,...,\rows} {
                \filldraw (\x*\spacing, \y*\spacing) circle (3 pt);
    
                \ifnum \x<\cols
                    \draw[thick] (\x*\spacing, \y*\spacing) -- ({(\x+1)*\spacing}, \y*\spacing);
                \fi
    
                \ifnum \y<\rows
                    \draw[thick] (\x*\spacing, \y*\spacing) -- (\x*\spacing, {(\y+1)*\spacing});
                \fi
    
                \ifnum \x=\cols
                    \draw[dashed] (\x*\spacing, \y*\spacing) -- ({(\x+\extend)*\spacing}, \y*\spacing);
                \fi
                
                \ifnum \y=\rows
                    \draw[dashed] (\x*\spacing, \y*\spacing) -- (\x*\spacing, {(\y+\extend)*\spacing});
                \fi
    
                \ifnum \x=1
                    \draw[dashed] (\x*\spacing, \y*\spacing) -- ({(\x-\extend)*\spacing}, \y*\spacing);
                \fi
    
                \ifnum \y=1
                    \draw[dashed] (\x*\spacing, \y*\spacing) -- (\x*\spacing, {(\y-\extend)*\spacing});
                \fi
            }
        }

        \filldraw[Blue] ({3}, {3}) circle (3 pt);
        \filldraw[Blue] ({3}, {7}) circle (3 pt);
        \filldraw[Blue] ({7}, {3}) circle (3 pt);
        \filldraw[Blue] ({7}, {7}) circle (3 pt);

        \draw[dashed, Blue] ({3}, {3}) -- ({3}, {7});
        \draw[dashed, Blue] ({3}, {7}) -- ({7}, {7});
        \draw[dashed, Blue] ({7}, {7}) -- ({7}, {3});
        \draw[dashed, Blue] ({7}, {3}) -- ({3}, {3});

        \draw[dashed, Blue] ({3}, {3}) -- ({3}, {1});
        \draw[dashed, Blue] ({3}, {3}) -- ({1}, {3});
        \draw[dashed, Blue] ({7}, {3}) -- ({9}, {3});
        \draw[dashed, Blue] ({7}, {3}) -- ({7}, {1});
        \draw[dashed, Blue] ({3}, {7}) -- ({3}, {9});
        \draw[dashed, Blue] ({3}, {7}) -- ({1}, {7});
        \draw[dashed, Blue] ({7}, {7}) -- ({9}, {7});
        \draw[dashed, Blue] ({7}, {7}) -- ({7}, {9});

        \draw[black][thick] ({2}, {2}) -- ({1}, {2});
        \draw[black][thick] ({2}, {2}) -- ({2}, {1});
        \draw[black][thick] ({2}, {4}) -- ({1}, {4});
        \draw[black][thick] ({2}, {6}) -- ({1}, {6});
        \draw[black][thick] ({2}, {8}) -- ({1}, {8});
        \draw[black][thick] ({2}, {8}) -- ({2}, {9});
        \draw[black][thick] ({4}, {2}) -- ({4}, {1});
        \draw[black][thick] ({6}, {2}) -- ({6}, {1});
        \draw[black][thick] ({8}, {2}) -- ({8}, {1});
        \draw[black][thick] ({8}, {2}) -- ({9}, {2});
        \draw[black][thick] ({8}, {4}) -- ({9}, {4});
        \draw[black][thick] ({8}, {6}) -- ({9}, {6});
        \draw[black][thick] ({8}, {8}) -- ({9}, {8});
        \draw[black][thick] ({8}, {8}) -- ({8}, {9});
        \draw[black][thick] ({6}, {8}) -- ({6}, {9});
        \draw[black][thick] ({4}, {8}) -- ({4}, {9});

        \node[Blue] at (7.4,7.4) {$x'$};
        \node[LimeGreen] at (8.8,5.2) {$B_{x'}$};
    \end{tikzpicture}
    
    \caption{The Coarse-Graining Geometry Visualized in $d=2$.}
    \label{fig: 16 to 4 Blocking on the Square Lattice with PBCs}
\end{figure}

Now, we want to find local transition rules that describe how we map the spin configuration $s_{B_{x'}}$ to a spin $\sigma_{x'}$ where $B_{x'}$ is the cube corresponding to the vertex $x'$. In order to guarantee some physical relevance of the outcome, we want to put some restrictions towards the choice of these rules.
\begin{enumerate}
    \item[R1:] The set of rules should be invariant under symmetry transformations of the block configurations. These include rotations, reflections and the $\mathds{Z}_2$-action $(s_{B_{x'}},\sigma_{x'})\mapsto(-s_{B_{x'}},-\sigma_{x'})$.
    \item[R2:] The rules must have a probabilistic interpretation, i.e.\ for any input configuration $s$ we must obtain
    \begin{equation}
        t(\sigma_{x'}=+1|s_{B_{x'}}),t(\sigma_{x'}=-1|s_{B_{x'}})\in[0,1]\quad;\quad t(\sigma_{x'}=+1|s_{B_{x'}})+t(\sigma_{x'}=-1|s_{B_{x'}})=1.
    \end{equation}
    \item[R3:] $t$ must preserve the monotonicity of the Ising model, i.e.\footnote{Here $s_{B_{x'}}\leq\tilde s_{B_{x'}}$ should be understood pointwise, i.e.\ $s_{B_{x'}}\leq\tilde s_{B_{x'}}\iff s_{y'}\leq\tilde s_{y'}$ for all $y'\in B_{x'}$}
    \begin{equation}
        s_{B_{x'}}\leq\tilde s_{B_{x'}}\implies t(\sigma_{x'}=+1|s_{B_{x'}})\leq t(\sigma_{x'}=+1|\tilde s_{B_{x'}}).
    \end{equation}
\end{enumerate}
Once we have specified these transition probabilities $t(\sigma_{x'}|s_{B_{x'}})$ we define a product probability kernel
\begin{equation}
    T=\prod_{x'\in\mathds{Z}^2}t(\sigma_{x'}| s_{B_{x'}}).
\end{equation}
The corresponding RGT is defined as the map (note that we have used the same $t$ for all $x'\in\mathds{Z}^2$ such that this map preserves translation invariance)
\begin{equation}
    \mu\mapsto T\mu,\qquad T\mu(\mathcal{A}):=\int_\Omega T(\mathcal{A}|s)\mu(\mathrm{d}s).
\end{equation}
One common example for $t$ is a majority rule vote given by
\begin{equation}
    t(\sigma_{x'}=+1|s_{B_{x'}})=\begin{cases}
        1&\sum_{x\in B_{x'}}s_x>0\\
        1/2&\sum_{x\in B_{x'}}s_x=0\\
        0&\sum_{x\in B_{x'}}s_x<0
    \end{cases}.
\end{equation}
Traditionally, one would then start to investigate properties of this map which led, among other things, to the discovery of the pathologies mentioned in the introduction. However, we want to pursue a different route and start considering, how we can make this concept work for finite graphs. Remember, one of the key advantages of finite graphs is that measures and interactions are related by the diffeomorphism $\mu\propto e^{-H}$ such that the van Enter et al.\ pathologies cannot appear on finite graphs. Although, we do avoid these pathologies by working with finite graphs, there seems to be a much bigger problem instead. Namely, on finite graphs implementing a coarse-graining necessarily reduces the number of vertices. Thus, any RGT, set up as explained above, can only be implemented on finite graphs as a map
\begin{equation}
    T:\mathscr{M}_1(\Omega_G/S)\to\mathscr{M}_1(\Omega_{G'}/S').
\end{equation}
where $|G|>|G'|$. This map consequently cannot exhibit fixed points nor can it be iterated. Therefore, we require a new idea in order to obtain a framework akin to Wilson's RG.\medskip

This is where the idea of inverse limits of models on finite graphs enters the story. Let $\Lambda_\infty$ be some $d$-dimensional lattice (like $\mathds{Z}^d$) and $T$ some block-spin RGT on $\Lambda_\infty$ with blocks $B_{x'}$ of size $b>1$. We characterize the block-structure via a map
\begin{equation}
    \phi:\Lambda_\infty\to\Lambda_\infty,\quad x\to x'\iff x\in B_{x'}\,.
\end{equation}
In order to set up an inverse limit, we choose a sequence $(\Lambda_n)_{n\in\mathds{N}}$ of finite, connected graphs satisfying:
\begin{enumerate}
    \item[i)] There exist coverings $\pi_{\infty n}:\Lambda_\infty\to\Lambda_n$ compatible with the block structure of $T$, i.e.\ there exist maps $\phi_n:\Lambda_n\to\Lambda_{n-1}$ satisfying
    \begin{equation}
        \pi_{\infty m}\circ\phi=\phi_n\circ\pi_{\infty n}\,.
    \end{equation}
    \item[ii)] There exist coverings $\pi_{nm}:\Lambda_n\to\Lambda_m$ of degree $b^{n-m}$ satisfying $\pi_{km}\pi_{nk}=\pi_{nm}$ as well as $\pi_{nm}\pi_{\infty n}=\pi_{\infty m}$.
\end{enumerate}

Once we have constructed such a system, we can define the finite-volume kernels $T_n,n\in\mathds{N}$ via
\begin{equation}
    T_n=\prod_{x\in\Lambda_{n-1}}t(\sigma_{x}|s_{\phi_n^{-1}(x)}).
\end{equation}
Let us denote the spaces of functions as\footnote{Recall that, in order to be precise, we have to consider the coset $\Iso(\Lambda_\infty)/\mathcal{D}(\pi_{\infty n})$. Also, in some scenarios the RGT can break some of the symmetries. One example would be a block-spin transformation on the $2\mathrm{D}$ triangular lattice using single triangles as blocks which breaks the $60^\circ$ rotational symmetry of the lattice down to $120^\circ$ symmetry. In that case, we would choose the symmetry group to be $S_n=S_\infty/\mathcal{D}(\pi_{\infty n})\times\mathds{Z}^2$ where $S_\infty\subseteq\Iso(\Lambda_\infty)$ is the largest possible symmetry group preserved by the RGT. However, in order for these definitions to be sensible it is important that the RGT preserves translation invariance in any case.} $\mathcal{V}_n:=\mathcal{F}(\Omega_{\Lambda_n}/\Iso(\Lambda_\infty)\times\mathds{Z}_2)$. The $T_n$ then allow us to construct a map $R_n$ acting on $\mathcal{V}_n$ via
\begin{equation}
    R_n:\mathcal{V}_n\to \mathcal{V}_{n-1},\,\, H_n\mapsto H_n',\qquad e^{-R_nH_n[\sigma]}:=\sum_{s\in\Omega_n}T_n(\sigma|s)e^{-H_n[s]}.
\end{equation}
While this provides a proper definition of $R_n$ acting onto $\mathcal{V}_n$, it is not yet sufficient for an implementation of Wilson's RG. The two obvious problems are 1) $R_n:\mathcal{V}_n\to \mathcal{V}_{n-1}$ cannot exhibit proper fixed points, as the domain and image live in separate spaces, and 2) By iterating the RGT we will end up on smaller and smaller lattices, eventually running out of vertices. The key idea is that we can avoid these issues by considering inverse limits. In order to obtain a useful definition of an inverse limit, we need linear projection maps $\omega_{nm}:\mathcal{V}_n\to \mathcal{V}_m$ that commute, i.e.\ satisfy $\omega_{km}\omega_{nk}=\omega_{nm}$. In that case, we define the space of inverse limits corresponding to the projections $\omega_{nm}$ as the space of sequences $v_n\in \mathcal{V}_n,n\in\mathds{N}$ satisfying 
\begin{equation}
    \omega_{nm}v_n=v_m.
\end{equation}
At first sight, it may seem that the natural choice for $\omega_{nm}$ is the pushforward map
\begin{equation}
    (\pi_{nm})_*:\mathcal{V}_n\to \mathcal{V}_m,\quad (\pi_{nm})_*v_n[\sigma]:=v_n[\pi_{nm}^*\sigma].
\end{equation}
The pullback on configurations is defined as $(\pi_{nm}^*\sigma)_x:=\sigma_{\pi_{nm}(x)}$. However, this is not quite the correct physically reasonable choice. Let us, for example, consider the nearest-neighbour Ising model. The total energy of the model scales with the volume $|\Lambda_n|$ while the local energy density (energy divided by volume) is the quantity that remains the same over all lattice sizes. However, the map $(\pi_{nm})_*$ implements a relation that matches the total energy among the different lattices rather than the energy density. We can account for this by adjusting the definition to
\begin{equation}
    \omega_{nm}:=\frac{1}{b^{n-m}}(\pi_{nm})_*.
\end{equation}
This then allows us to define the space of inverse limits as
\begin{equation}
    \mathcal{V}_\infty:=\{(v_n)_{n\in\mathds{N}}\,:\,\omega_{nm}v_n=v_m\,\,\forall n,m\in\mathds{N}\}.
\end{equation}
Since the projections $\omega_{nm}$ are linear, $\mathcal{V}_\infty$ is again a vector space. The key question is how we can extend the action of $T$ onto $\mathcal{V}_\infty$. If the finite-volume RGTs $R_n$ and the projections $\omega_{nm}$ would commute, i.e.\ we would find $\omega_{nn-1}R_{n+1}=R_n\omega_{n+1n}$, then we could simply set $(Rv)_n:=R_{n+1}v_{n+1}$ obtaining a new compatible sequence. However, from numerical simulations (as well as analytical arguments) we know that this is not true, i.e.\ the RGT is affected by boundary effects that alter the result depending on whether we project first and then block or vice versa. Fortunately, there is an alternative way to obtain a compatible sequence that does not require this strict commutation property. Let $v\in \mathcal{V}_\infty$ such that the limit
\begin{equation}
    \lim_{n\to\infty}\omega_{n-1m}R_nv_n=:w_m
\end{equation}
exists for all $m\in\mathds{N}$. Then the sequence $w_m$ is again compatible as
\begin{equation}
    \omega_{km}w_k=\omega_{km}\lim_{n\to\infty}\omega_{n-1k}R_nv_n=\lim_{n\to\infty}\omega_{n-1m}R_nv_n=w_m.
\end{equation}
In that case the action of $T$ onto $v$ is well-defined via $(Rv)_n:=w_n$ for all $n\in\mathds{N}$. Note that $\omega_{n-1m}R_nv_n\in\mathcal{V}_m$ for all $n>m$. Thus, we test convergence in finite dimensional vector spaces. The question therefore is, for which points $v\in \mathcal{V}_\infty$ the sequences $\lim_{n\to\infty}\omega_{n-1m}R_nv_n$ all converge such that they belong to the domain
\begin{equation}
    \mathscr{D}(R):=\left\{v\in \mathcal{V}_\infty\,:\,\lim_{n\to\infty}\omega_{n-1m}R_nv_n\,\,\mathrm{exists}\,\,\forall m\in\mathds{N}\right\}\subset \mathcal{V}_\infty.
\end{equation}
Alternatively, we would like to find some set of "physically interesting" points $\mathcal{V}_\mathrm{phys}\subseteq\mathscr{D}(R)$ that satisfies $R(\mathcal{V}_\mathrm{phys})\subseteq \mathcal{V}_\mathrm{phys}$. We will tackle this task in the next section. Before we get there, it is about time to consider an example.

\begin{example}
    Let us consider $\Lambda_\infty=\mathds{Z}^2$. The block-structure is given by
    \begin{equation}
        \phi:\mathds{Z}^2\to\mathds{Z}^2,\qquad \phi(i,j):=(\lfloor i/2\rfloor,\lfloor j/2\rfloor).
    \end{equation}
    This leads to $2\times2$-blocks matching the visualization in Figure \ref{fig: 16 to 4 Blocking on the Square Lattice with PBCs} on the previous page. The finite graphs are periodic square lattices $\Lambda_n$ with vertex sets $\mathds{Z}_{2^n}\times\mathds{Z}_{2^n}$. The coverings are given by\footnote{There will be some abuse of notation as we will use $\pi$ both for the entire graph covering as well as only the vertex map of the covering like in this case.}
    \begin{equation}
        \pi_{\infty n}(i,j):=(i\,\,\mathrm{mod}\,\,2^n,j\,\,\mathrm{mod}\,\, 2^n)\qquad;\qquad \pi_{nm}(i,j):=(i\,\,\mathrm{mod}\,\, 2^m,j\,\,\mathrm{mod}\,\,2^m)
    \end{equation}
    with $\deg(\pi_{nm})=b^{n-m},b=4$. We can visualize the lattices as finer and finer discretization of a torus $T$. The finite-size block-structures are simply given by
    \begin{equation}
        \phi_n(i,j):=(\lfloor i/2\rfloor,\lfloor j/2\rfloor).
    \end{equation}
    On this geometry, we could then, for example, set up the previously mentioned majority rule transformation.
\end{example}

\subsection{Main Results}

\begin{lemma}\label{Lem1}
    Let $(a_n)_{n\in\mathds{N}}\subset\mathds{R}^d$ be a sequence such that\footnote{By equivalence of norms on $\mathds{R}^d$ we do not have to specify the choice of norm.}
    \begin{equation}
        \sum_{n=1}^\infty\|a_n-a_{n-1}\|<\infty.
    \end{equation}
    Then the sequence converges to some limit $a=\lim_{n\to\infty}a_n$.
\end{lemma}

\begin{proof}
    Let us consider the sequence $S_n:=\sum_{k=1}^n\|a_k-a_{k-1}\|$. Since $(S_n)_{n\in\mathds{N}}\subset\mathds{R}$ is increasing and bounded it converges to some limit $\lim_{n\to\infty}S_n=M$. Therefore, for any $\varepsilon>0$ exists some $N\in\mathds{N}$ such that $|M-S_n|<\varepsilon$ for all $n\geq N$. Consequently, by applying the triangle inequality, we conclude that
    \begin{equation}
        \|a_n-a_m\|\leq\sum_{k=m+1}^n\|a_k-a_{k-1}\|\leq M-S_m<\varepsilon.
    \end{equation}
    for all $n\geq m\ge N$. Therefore, the sequence $(a_n)_{n\in\mathds{N}}$ is Cauchy and hence, by completeness of $\mathds{R}^d$, it converges.
\end{proof}

\begin{lemma}\label{Lem3}
    We find
    \begin{equation}
        \|R_nG_n-R_nH_n\|_\infty\leq\|G_n-H_n\|_\infty.
    \end{equation}
\end{lemma}

\begin{proof}
    Let us consider $G_n=H_n+(G_n-H_n)$. We shall compute
    \begin{equation}
        R_nG_n(\sigma)=-\ln\left(\sum_sT(\sigma|s)e^{-H_n(s)-(G_n-H_n)(s)}\right).
    \end{equation}
    Since $\exp$ as well as $T(\sigma|s)$ are non-negative, it follows that
    \begin{equation}
        e^{-\|G_n-H_n\|_\infty}\sum_sT(\sigma|s)e^{-H_n(s)}\leq\sum_sT(\sigma|s)e^{-H_n(s)-(G_n-H_n)(s)}\leq e^{\|G_n-H_n\|}\sum_sT(\sigma|s)e^{-H_n(s)}.
    \end{equation}
    These bounds together with the fact that the logarithm is increasing then gives us
    \begin{align}
        R_nH_n(\sigma)-\|G_n-H_n\|_\infty\leq R_nG_n(\sigma)\leq R_nH_n(\sigma)+\|G_n-H_n\|_\infty\nonumber\\
        \implies \|R_nG_n-R_nH_n\|_\infty\leq\|G_n-H_n\|_\infty.
    \end{align}
\end{proof}

Let $H\in \mathcal{V}_\infty$. At first we want to isolate the part of $H$ that troubles the commutation between $R_n$ and $\omega_{nm}$. Therefore, we want to partition $\Lambda_n$ into fundamental domains $Q_1^{n},...,Q_b^{n}$ such that $\phi_n^{-1}(x)\subset Q_i^n$ for all $x\in Q_i^{n-1}$. In words, we demand that the RG maps the spins in $Q_i^n$ exactly into $Q_i^{n-1}$. This allows us to define the defect at level $n$ via
\begin{equation}\label{Eq3}
    \Delta H_n(\sigma):=H_n(\sigma)-\sum_{i=1}^bH_{n-1}(\sigma_{Q_i^n}).
\end{equation}
Recall, that by definition of the fundamental domains, the map $(\pi_{nn-1})_V|_{V(Q_i^n}:V(Q_i^n)\to V(\Lambda_{n-1})$ is a bijection. Therefore, it is reasonable to set
\begin{equation}
    \sigma_{Q_i^n}\in \Omega_{\Lambda_{n-1}},\quad (\sigma_{Q_i^n})_x:=\sigma_{(\pi_{nn-1})_V|_{V(Q_i^n)}^{-1}(x)}.
\end{equation}

\begin{lemma}\label{Lem4}
    Let $H\in\mathcal{V}_\infty$. Then we find that
    \begin{equation}
        \|\omega_{n-1m}R_nH_n-R_{m+1}H_{m+1}\|_\infty\leq\sum_{k=m+2}^n\frac{1}{b^{k-1-m}}\|\Delta H_k\|_\infty.
    \end{equation}
\end{lemma}

\begin{proof}
    We want to prove this by induction on the difference of $n$ and $m$. Let us first assume that $n-m=2$. In that case, we shall compute
    \begin{align}
        \omega_{n-1m}R_nH_n(\sigma)&=-\frac{1}{b}\ln\left(\sum_{s\in\Omega_n}T_n(\pi_{n-1m}^*\sigma|s)e^{-H_n(s)}\right)\nonumber\\
        &=-\frac{1}{b}\ln\left(\sum_{s\in\Omega_n}T_n(\pi_{n-1m}^*\sigma|s)e^{-\sum_{i=1}^bH_{n-1}(s_{Q_i^n})-\Delta H_n(s)}\right).
    \end{align}
    Since the kernel $T_n$ is simply a product over blocks contained in one of the $Q_i^n$, we find that
    \begin{align}
        \ln\left(\sum_{s\in\Omega_n}T_n(\pi_{n-1m}^*\sigma|s)e^{-\sum_{i=1}^bH_{n-1}(s_{Q_i^n})}\right)&=\ln\left(\prod_{i=1}^b\sum_{s\in\Omega_{n-1}}T_{n-1}\left((\pi_{n-1m}^*\sigma)_{Q_i^{n-1}}|s\right)e^{-H_{n-1}(s)}\right)\nonumber\\
        &=\sum_{i=1}^b\ln\left(\sum_{s\in\Omega_{n-1}}T_{n-1}\left((\pi_{n-1m}^*\sigma)_{Q_i^{n-1}}|s\right)e^{-H_{n-1}(s)}\right).
    \end{align}
    However, we note that $(\pi_{n-1m}^*\sigma)_{Q_i^{n-1}}=\sigma$ up to a potential translation. By translation invariance of $R_{n-1}$ we therefore conclude that
    \begin{equation}
        -\ln\left(\sum_{s\in\Omega_n}T_n(\pi_{n-1m}^*\sigma|s)e^{-\sum_{i=1}^bH_{n-1}s_{Q_i^n}}\right)=bR_{n-1}H_{n-1}(\sigma).
    \end{equation}
    and therefore $\omega_{n-1m}R_nH_n^Q=R_{n-1}H_{n-1}$ where $H_{n}^Q\in\mathcal{V}_n$ denotes the Hamilton function $H_{n}^Q(s):=\sum_{i=1}^bH_{n-1}(s_{Q_i^n})$.
    This implies that
    \begin{equation}
        \|\omega_{n-1m}R_nH_n-R_{n-1}H_{n-1}\|_\infty=\|\omega_{n-1m}(R_nH_n-R_nH_n^Q)\|_\infty=\frac{1}{b}\left\|R_nH_n-R_nH_{n-1}^Q\right\|_\infty.
    \end{equation}
    At last, we may use Lemma \ref{Lem3} giving us
    \begin{equation}
        \left\|R_nH_n-R_nH_{n}^Q\right\|_\infty\leq\|H_n-H_{n}^Q\|_\infty=\|\Delta H_n\|_\infty.
    \end{equation}
    Now, for arbitrary difference $n-m\geq2$ we may use a telescoping argument. Precisely, we find that
    \begin{align}
        \|\omega_{n-1m}R_nH_n-R_{m+1}H_{m+1}\|_\infty&\leq\sum_{k=m}^{n-2}\|\omega_{km}(\omega_{k+1k} R_{k+2}H_{k+2}-R_{k+1}H_{k+1})\|_\infty\nonumber\\
        &\leq\sum_{k=m}^{n-2}\frac{1}{b^{k-m}}\|\omega_{k+1k}R_{k+2}H_{k+2}-R_{k+1}H_{k+1}\|_\infty\nonumber\\
        &\leq\sum_{k=m}^{n-2}\frac{1}{b^{k+1-m}}\|\Delta H_{k+2}\|_\infty.
    \end{align}
\end{proof}

This result confirms our intuition that $\Delta H_n$ is the part of the Hamilton function that is responsible for the non-commutation of the RGT and the projections. The first application of this result is the following. Let us compute
\begin{equation}
    \|\omega_{n-1m}R_nH_n-\omega_{n-2m}R_{n-1}H_{n-1}\|_\infty\leq\frac{1}{b^{n-2-m}}\|\omega_{n-1n-2}R_nH_n-R_{n-1}H_{n-1}\|_\infty\leq\frac{b^{m+1}\|\Delta H_n\|_\infty}{b^{n}}
\end{equation}
In case that the sum
\begin{equation}
    \sum_{n=m+1}^\infty\frac{\|\Delta H_n\|_\infty}{b^n}
\end{equation}
is finite, the sequence $\omega_{n-1m}R_nH_n$ satisfies the condition from Lemma \ref{Lem1} and hence it converges. This motivates the following theorem:

\begin{theorem}\label{Thm1}
    The subspace
    \begin{equation}
        \mathcal{V}_\mathrm{phys}:=\left\{v\in \mathcal{V}_\infty\,:\,\exists C>0,1>\alpha\geq0,\,\,\left\|v_n(\sigma)-\sum_{i=1}^bv_{n-1}(\sigma_{Q_i})\right\|_\infty\leq Cb^{\alpha n}\right\}\subset \mathcal{V}_\infty
    \end{equation}
    satisfies $\mathcal{V}_\mathrm{phys}\subset\mathscr{D}(R)$.
\end{theorem}

\begin{proof}
    Let $H\in \mathcal{V}_\mathrm{phys}$. We quickly compute
    \begin{equation}
        \sum_{n=m+1}^\infty\frac{\|\Delta H_n\|_\infty}{b^n}\leq\sum_{n=m+1}^\infty Cb^{(\alpha-1)n}.
    \end{equation}
    Since $\alpha-1<0$ this is simply a geometric series and we find
    \begin{equation}
        \sum_{n=m+1}^\infty Cb^{(\alpha-1)n}=\frac{Cb^{(\alpha-1)(m+1)}}{1-b^{\alpha-1}}=\frac{Cb^{m(\alpha-1)+\alpha}}{b-b^\alpha}<\infty.
    \end{equation}
\end{proof}

With this theorem we obtain a partial answer to the questions we stated at the end of the previous section. In order to obtain a satisfactory answer we also want to establish that $R(\mathcal{V}_\mathrm{phys})\subseteq \mathcal{V}_\mathrm{phys}$ and that it contains at least a part of the physically interesting models.

\begin{theorem}\label{Thm2}
    The subspace $\mathcal{V}_\mathrm{phys}$ is invariant under the RGT, i.e.\ $R(\mathcal{V}_\mathrm{phys})\subseteq \mathcal{V}_\mathrm{phys}$.
\end{theorem}

\begin{proof}
    Let $H\in \mathcal{V}_\phys$. We set $G=RH\in \mathcal{V}_\infty$\footnote{$G$ is well-defined by Theorem \ref{Thm1}.}. Our aim is to bound the difference
    \begin{equation}
        \Delta G_m(\sigma)=G_m(\sigma)-\sum_{i=1}^bG_{m-1}(\sigma_{Q_i^m}).
    \end{equation}
    By inserting the definition of $R$, this can be written as
    \begin{equation}
        \Delta G_m(\sigma)=\lim_{n\to\infty}\left[\omega_{n-1m}R_nH_n(\sigma)-\sum_{i=1}^b\omega_{n-2m-1}R_{n-1}H_{n-1}(\sigma_{Q_i^m})\right].
    \end{equation}
    By Lemma \ref{Lem4} we find that
    \begin{equation}
        \|\omega_{n-1m}R_nH_n-R_{m+1}H_{m+1}\|_\infty\leq\sum_{k=m+2}^n\frac{1}{b^{k-1-m}}\|\Delta H_k\|_\infty.
    \end{equation}
    Likewise, we find that
    \begin{equation}
        \|\omega_{n-2m-1}R_{n-1}H_{n-1}-R_mH_m\|_\infty\leq\sum_{k=m+1}^{n-1}\frac{1}{b^{k-m}}\|\Delta H_k\|.
    \end{equation}
    Using the triangular inequality, we obtain the following bound:
    \begin{align}
        \left|\omega_{n-1m}R_nH_n(\sigma)-\sum_{i=1}^b\omega_{n-2m-1}R_{n-1}H_{n-1}(\sigma_{Q_i^m})\right|\leq\sum_{k=m+2}^n\frac{1}{b^{k-1-m}}\|\Delta H_k\|_\infty\nonumber\\+b\sum_{k=m+1}^{n-1}\frac{1}{b^{k-m}}\|\Delta H_k\|
        +\left|R_{m+1}H_{m+1}(\sigma)-\sum_{i=1}^bR_mH_m(\sigma_{Q_i^{m+1}})\right|.
    \end{align}
    For the remaining term, we once more use that the defect measures by how much the application of the RGT at level $m+1$ differs from applying it at level $m$, i.e.
    \begin{equation}
        R_{m+1}H_{m+1}(\sigma)\in\left[\sum_{i=1}^bR_mH_m(\sigma_{Q_i^{m+1}})-\|\Delta H_{m+1}\|_\infty,\sum_{i=1}^bR_mH_m(\sigma_{Q_i^{m+1}})-\|\Delta H_{m+1}\|_\infty\right].
    \end{equation}
    Thus, the whole bound is given by
    \begin{align}
        \left|\omega_{n-1m}R_nH_n(\sigma)-\sum_{i=1}^b\omega_{n-2m-1}R_{n-1}H_{n-1}(\sigma_{Q_i^m})\right|\leq\sum_{k=m+2}^n\frac{1}{b^{k-1-m}}\|\Delta H_k\|_\infty\nonumber\\+b\sum_{k=m+1}^{n-1}\frac{1}{b^{k-m}}\|\Delta H_k\|
        +\|\Delta H_{m+1}\|_\infty.
    \end{align}
    Now, we may use the fact that $H\in\mathcal{V}_\phys$ and therefore $\|\Delta H_n\|_\infty\leq Cb^{\alpha n}$ for some $\alpha\in[0,1)$. This allows us to compute
    \begin{align}
        \sum_{k=m+2}^n\frac{1}{b^{k-1-m}}\|\Delta H_k\|_\infty+b\sum_{k=m+1}^{n-1}\frac{1}{b^{k-m}}\|\Delta H_k\|
        +\|\Delta H_{m+1}\|_\infty\nonumber\\
        \leq C\left(b^{m+1}\sum_{k=m+2}^nb^{(\alpha-1)k}+b^{m+1}\sum_{k=m+1}^{n-1}b^{(\alpha-1)k}+b^{\alpha(m+1)}\right)\nonumber\\
        =C\left(\frac{b^{m+1}\left(b^{(m+2)(\alpha-1)}-b^{(n+1)(\alpha-1)}\right)}{1-b^{\alpha-1}}+\frac{b^{m+1}\left(b^{(m+1)(\alpha-1)}-b^{n(\alpha-1)}\right)}{1-b^{\alpha-1}}+b^{\alpha(m+1)}\right).
    \end{align}
    In the limit $n\to\infty$ the terms containing $b^{n(\alpha-1)}$ vanish since $b>1$ and $(\alpha-1)<0$. Thus, we conlcude that
    \begin{equation}
        \|\Delta G_m\|_\infty\leq C\left(\frac{b^\alpha(1+b^{\alpha-1})}{1-b^{\alpha-1}}+b^\alpha\right)b^{\alpha m}.
    \end{equation}
    Hence the renormalized Hamilton function $G$ admits a bound
    \begin{equation}
        \|\Delta G_m\|_\infty\leq\tilde Cb^{\alpha m},\quad\tilde C=C\left(\frac{b^\alpha(1+b^{\alpha-1})}{1-b^{\alpha-1}}+b^\alpha\right)
    \end{equation}
    and therefore $G\in \mathcal{V}_\phys$.
\end{proof}

Theorem \ref{Thm1} and \ref{Thm2} together prove that $R:\mathcal{V}_\phys\to \mathcal{V}_\phys$ is well-defined. In other words, on $\mathcal{V}_\phys$ we can implement a pathology free version of the RGT. Wilson also proposed that the renormalization map should be differentiable. Here, we want to establish that $R$ is at least continuous. In order to do so, we must first equip $\mathcal{V}_\phys$ with a suitable norm $\|\cdot\|$. By leveraging the decomposition from Equation (\ref{Eq3}) once more, we may compute (recall that $|\Lambda_n|=b\cdot|\Lambda_{n-1}|$)
\begin{equation}
    \frac{\|H_n\|_\infty}{|\Lambda_n|}\leq\frac{1}{b}\sum_{i=1}^b\frac{\|H_{n-1}\|_\infty}{|\Lambda_{n-1}|}+\frac{\|\Delta H_n\|_\infty}{|\Lambda_n|}=\frac{\|H_{n-1}\|_\infty}{|\Lambda_{n-1}|}+\frac{\|\Delta H_n\|_\infty}{|\Lambda_n|}.
\end{equation}
Iterating this procedure implies the bound
\begin{equation}
    \frac{\|H_n\|_\infty}{|\Lambda_n|}\leq\frac{\|H_0\|_\infty}{|\Lambda_0|}+\sum_{k=1}^n\frac{\|\Delta H_k\|}{|\Lambda_k|}.
\end{equation}
We conclude that for any $H\in \mathcal{V}_\phys$ the sequence $\|H_n\|_\infty/|\Lambda_n|$ is bounded above from 
\begin{equation}
    \frac{\|H_0\|_\infty}{|\Lambda_0|}+\sum_{k=1}^\infty\frac{\|\Delta H_k\|}{|\Lambda_k|}<\infty.
\end{equation}
Further, it is also monotonously increasing as
\begin{equation}
    \frac{H_n(\pi_{nn-1}^*\sigma)}{|\Lambda_n|}=\frac{H_{n-1}(\sigma)}{|\Lambda_{n-1}|}\quad\forall\sigma\in\Omega_{\Lambda_{n-1}}\implies\frac{\|H_n\|_\infty}{|\Lambda_n|}\geq\frac{\|H_{n-1}\|_\infty}{|\Lambda_{n-1}|}
\end{equation}
and hence it converges. This motivates the definition
\begin{equation}
    \|H\|:=\lim_{n\to\infty}\frac{\|H_n\|_\infty}{|\Lambda_n|}.
\end{equation}

\begin{theorem}\label{Thm3}
    The map $R:\mathcal{V}_\phys\to \mathcal{V}_\phys$ is continuous regarding the norm $\|\cdot\|$ on $\mathcal{V}_\phys$.
\end{theorem}

\begin{proof}
    Let $H,G\in \mathcal{V}_\phys$ and $\varepsilon>0$. We may compute
    \begin{equation}
        \|\omega_{n-1m}(R_nG_n-R_nH_n)\|_\infty\leq\frac{1}{b^{n-1-m}}\|R_nG_n-R_nH_n\|_\infty.
    \end{equation}
    By Lemma \ref{Lem3} we find that $\|R_nG_n-R_nH_n\|_\infty\leq\|G_n-H_n\|_\infty$. Furthermore, we may use that $\|G_n-H_n\|_\infty\leq b^n\|G-H\|$ giving us
    \begin{equation}
        \|\omega_{n-1m}(R_nG_n-R_nH_n)\|_\infty\leq\frac{1}{b^{n-1-m}}b^n\|G-H\|=b^{m+1}\|G-H\|.
    \end{equation}
    It follows that $R$ is Lipschitz continuous with constant $b$ and hence continuous.
\end{proof}

Thus, we successfully established that $R:\mathcal{V}_\phys\to \mathcal{V}_\phys$ is not only well-defined but also continuous regarding a very natural norm as $\|\cdot\|$ is a bound on the maximal possible energy density. However, while it is nice to have a framework which is mathematically sound, we still need to answer the question whether this framework is physically relevant. There are two subquestions to this:
\begin{enumerate}
    \item[i)] Does $\mathcal{V}_\phys$ contain any physically interesting points?
    \item[ii)] Does this framework keep a connection to the physics on the infinite lattice $\Lambda_\infty$? 
\end{enumerate}
We will try to answer them one by one starting with i). In \cite{van_Enter_1993} the authors spend quite some time discussing different possible definitions for physically relevant interactions on the infinite lattice. An interaction in this context is a collection of maps
\begin{equation}
    \Phi=\{\Phi_\Gamma:\Omega_\Gamma\to\mathds{R}\,:\,\Gamma\subset\Lambda_\infty,|\Gamma|<\infty\}.
\end{equation}
The most natural space for the theory of infinite volume Gibbs measures (see also \cite{friedli_velenik_2017}) is the space of absolutely summable and translation-invariant interactions denoted as $\mathcal{B}^1$. An interaction $\Phi$ is called absolutely summable if
\begin{equation}
    \|\Phi\|_{\mathcal{B}^1}:=\sum_{\substack{\Gamma\subset\Lambda_\infty,|\Gamma|<\infty\\0\in\Gamma}}\|\Phi_\Gamma\|_\infty<\infty.
\end{equation}
In our setting, we do not only want to impose translation invariance but invariance under the full symmetry group\footnote{Recall that in case the RGT breaks some of the symmetries, we shall restrict ourselves to the maximal possible symmetry group preserved by the RGT.} $\Iso(\Lambda_\infty)\times\mathds{Z}_2$ which requires that $\Phi_\Gamma(-\sigma_\Gamma)=\Phi_\Gamma(\sigma_\Gamma)$ as well as 
\begin{equation}
    \Phi_\Gamma=\Phi_{g(\Gamma)}\circ(g^{-1}|_\Gamma)^*\qquad\forall g\in\Iso(\Lambda_\infty).
\end{equation}
The following proposition relates these infinite volume interactions to points in $\mathcal{V}_\infty$.

\begin{proposition}
    Let $\Phi$ be an absolutely summable interaction and $\Iso(\Lambda_\infty)\times\mathds{Z}^2$-invariant interaction on $\Lambda_\infty$. Then there is a canonical way to construct a corresponding $H^\Phi\in \mathcal{V}_\infty$. In case that $\Phi$ satisfies the stricter inequality
    \begin{equation}\label{Eq1}
        \sum_{\substack{\Gamma\subset\Lambda_\infty,|\Gamma|<\infty\\0\in\Gamma}}\diam(\Gamma)\cdot\|\Phi_\Gamma\|_\infty<\infty,
    \end{equation}
    we find that $H^\Phi\in \mathcal{V}_\phys$.
\end{proposition}

\begin{proof}
    Let us consider the coset $S_n=\{\Gamma\subset\Lambda_\infty\,:\,|\Gamma|<\infty\}/\mathcal{D}(\pi_{\infty n})$. We then define
    \begin{equation}
        H^{\Phi}_n(\sigma):=\sum_{[\Gamma]\in S_n}\Phi_\Gamma(\pi_{\infty n}^*\sigma).
    \end{equation}
    First, we quickly note that each map is well-defined as 
    \begin{equation}
        \sum_{[\Gamma]\in S_n}|\Phi_\Gamma(\pi_{\infty n}^*\sigma)|\leq\|\Phi\|_{\mathcal{B}^1}\cdot|\Lambda_n|<\infty
    \end{equation}
    and thus the series $\sum_{[\Gamma]\in S_n}\Phi_\Gamma(\pi_{\infty n}^*\sigma)$ converges. The next step is to prove that $H^\Phi\in \mathcal{V}_\infty$, i.e.\ that $\omega_{nm}H_n^\Phi=H_m^\Phi$. Let us therefore consider
    \begin{equation}
        H_n^\Phi(\pi_{nm}^*\sigma)=\sum_{[\Gamma]\in S_n}\Phi_\Gamma(\pi_{\infty m}^*\sigma).
    \end{equation}
    Here, we note that for each equivalence class $[\Gamma]\in S_m$ there exist equivalence classes $[\Gamma_1],...,[\Gamma_{b^{n-m}}]\in S_n$ such that
    \begin{equation}
        [\Gamma]=\bigcup_{i=1}^{b^{n-m}}[\Gamma_i].
    \end{equation}
    Further, we find that $(\pi_{\infty m}^*\sigma)_{\Gamma_i}=(\pi_{\infty m}^*\sigma)_\Gamma$. Thus, we conclude that
    \begin{equation}
        H_n^\Phi(\pi_{nm}^*\sigma)=b^{n-m}\sum_{[\Gamma]\in S_m}\Phi_\Gamma(\pi_{\infty m}^*\sigma)=b^{n-m}H_{m}^\Phi(\sigma)\implies\omega_{nm}H_n^\Phi=H_m^\Phi\,.
    \end{equation}
    For the second statement, let us assume that $\Phi$ satisfies (\ref{Eq1}). We then must prove that $\Delta H_n^\Phi$ satisfies the necessary bound such that $H^\Phi\in \mathcal{V}_\phys$. Let us start by considering
    \begin{equation}
        H_{n}^\Phi(\sigma)-\sum_{i=1}^bH_{n-1}^\Phi(\sigma_{Q_i})=\sum_{[\Gamma]\in S_n}\Phi_\Gamma(\pi_{\infty n}^*\sigma)-\sum_{i=1}^b\sum_{[\Delta]\in S_{n-1}}\Phi_\Delta(\pi_{\infty n-1}^*\sigma_{Q_i}).
    \end{equation}
    Let us now decompose $S_n=S^\loc_n\cup S_n^\bnd$ where
    \begin{equation}
        S_n^\loc:=\{[\Gamma]\in S_n\,:\,\pi_{\infty n}(\Gamma)\in Q_i,i\in\{1,...,b\}\},\qquad S_n^\bnd=S_n\backslash S_n^\loc.
    \end{equation}
    For any $[\Gamma]\in S_n^\loc$ there exists a corresponding $[\Delta(\Gamma)]\in S_{n-1}$ satisfying $\Phi_\Gamma(\pi_{\infty n}^*\sigma)=\Phi_{\Delta(\Gamma)}(\pi_{\infty n-1}^*\sigma_{Q_i})$. Therefore, we conclude that
    \begin{equation}
        \|\Delta _nH_n^\Phi\|_\infty\leq\sum_{[\Gamma]\in S_n^\bnd}\|\Phi_\Gamma\|_\infty=b\sum_{\substack{[\Gamma]\in S_n^\bnd\\\Gamma\cap \pi_{\infty n}^{-1}(Q_1)\ne\emptyset}}\|\Phi_\Gamma\|_\infty\,.
    \end{equation}
    For any $[\Gamma]\in S_n^\bnd,\Gamma\cap\pi_{\infty n}^{-1}(Q_1)\neq\emptyset$ there must exist at least two vertices $x\in Q_1,y\in \Lambda_n\backslash Q_1$ such that $\pi_{\infty n}^{-1}(x)\cap\Gamma\neq\emptyset,\pi_{\infty n}^{-1}(y)\neq\emptyset$. Consequently, we find that $\diam(\Gamma)\geq d(x,y)> d(x,\partial Q_1)$. Therefore, we may bound the previous sum via (by translation invariance we may choose $\tilde x$ to be any point in $\pi_{\infty n}^{-1}(x)$) 
    \begin{equation}
        \sum_{\substack{[\Gamma]\in S_n^\bnd\\\Gamma\cap\pi_{\infty n}^{-1}(Q_1)\ne\emptyset}}\|\Phi_\Gamma\|_\infty\leq\sum_{x\in Q_1}\sum_{\substack{\Gamma\ni\tilde x\\\diam(\Gamma)> d(x,\partial Q_i)}}\|\Phi_\Gamma\|_\infty=\sum_{r=0}^\infty\sum_{\substack{x\in Q_1\\d(x,\partial Q_1)=r}}\sum_{\substack{\Gamma\ni\tilde x\\\diam(\Gamma)>r}}\|\Phi_\Gamma\|_\infty\,.
    \end{equation}
    By translation invariance of the interaction, the last sum is independent of the choice of $x$, i.e.\ $\sum_{\substack{\Gamma\ni\tilde x\\\diam(\Gamma)>r}}\|\Phi_\Gamma\|_\infty=\sum_{\substack{\Gamma\ni0\\\diam(\Gamma)>r}}\|\Phi_\Gamma\|_\infty$. This allows us to consider the second sum in isolation. It is bounded by
    \begin{equation}
        \sum_{\substack{x\in Q_1\\d(x,\partial Q_i)=r}}\leq|\partial Q_1|\leq Cb^{\alpha n}.
    \end{equation}
    where $\alpha<1$ since the boundary scales strictly less than the volume. Thus, we obtain
    \begin{equation}
        \|\Delta H_n^\Phi\|_\infty\leq b\tilde C\cdot b^{\alpha n}\sum_{r=0}^\infty\sum_{\substack{\Gamma\ni0\\\diam(\Gamma)>r}}\|\Phi_\Gamma\|_\infty\,.
    \end{equation}
    For the remaining sum, we note that each $\Gamma$ appears exactly $\diam(\Gamma)$ times such that
    \begin{equation}
        \sum_{r=0}^\infty\sum_{\substack{\Gamma\ni0\\\diam(\Gamma)>r}}\|\Phi_\Gamma\|_\infty=\sum_{\Gamma\ni0}\diam(\Gamma)\cdot\|\Phi_\Gamma\|_\infty=M<\infty.
    \end{equation}
    Hence, we conclude
    \begin{equation}
        \|\Delta H_n^\Phi\|_\infty\leq bM\tilde C\cdot b^{\alpha n}\implies H^\Phi\in \mathcal{V}_\phys.
    \end{equation}
\end{proof}

The main takeaway of this result is that $\mathcal{V}_\phys$ contains all translation-invariant and absolutely summable infinite volume interactions that decay at least linearly with the diameter of the interaction. Of course, in most physical models, the interaction is either of finite range or decays exponentially which are well within the allowed region. Therefore, $\mathcal{V}_\phys$ contains practically all physically interesting models.\smallskip

The second question does not yet have a definite answer as the following key result is still missing.

\begin{conjecture}\label{Conj1}
    Let $\Phi\in\mathcal{B}^1$ such that $H^\Phi\in \mathcal{V}_\phys$. Further, let us assume that there exists a renormalized interaction $\Phi'\in\mathcal{B}^1$ for the infinite volume RGT $T$. Then we find that
    \begin{equation}
        R_nH_n^\Phi=H_{n-1}^{\Phi'}+\delta H_{n-1}
    \end{equation}
    where
    \begin{equation}
        \lim_{n\to\infty}\frac{\|\delta H_n\|_\infty}{|\Lambda_n|}=0.
    \end{equation}
\end{conjecture}

\begin{corollary}
    Let $\Phi\in\mathcal{B}^1$ such that there exists a renormalized interaction $\Phi'\in\mathcal{B}^1$. If further $H^\Phi\in \mathcal{V}_\phys$ then we find that
    \begin{equation}
        RH^\Phi=H^{\Phi'}.
    \end{equation}
\end{corollary}

\begin{proof}
    A quick computation shows that
    \begin{equation}
        (RH^\Phi)_m=\lim_{n\to\infty}\omega_{n-1m}R_nH_n^\Phi=\lim_{n\to\infty}\omega_{n-1m}\big(H_{n-1}^{\Phi'}+\delta H_{n-1}\big)=H_m^{\Phi'}+\lim_{n\to\infty}\frac{b^m\delta H_{n-1}\circ\pi_{n-1m}^*}{|\Lambda_{n-1}|}
    \end{equation}
    and hence
    \begin{equation}
        \|(RH^\Phi)_m-H_m^{\Phi'}\|_\infty\leq\lim_{n\to\infty}\frac{b_m\|\delta H_n\|_\infty}{|\Lambda_n|}=0\implies RH^\Phi=H^{\Phi'}.
    \end{equation}
\end{proof}

This would establish that in the pathology-free region of the infinite volume RGT $R_\infty$, the inverse-limit RGT $R$ coincides with $R_\infty$. This is of course a property one should expect a pathology-free extension of $R_\infty$ to have in order for it to be a physically relevant definition. It hinges entirely on Conjecture \ref{Conj1} which simply put tells us that for infinite volume interactions $\Phi$ that admit a renormalized interaction $\Phi'$, we can approximate $\Phi'$ by finite volume RGTs up to an error that scales like $o(|\Lambda|)$. The intuition behind this conjecture is that the RGT acts locally on the lattice and thus, although it can create arbitrarily large interactions out of even finite models, these decay fast enough such that their effects are not felt strongly in finite but large volumes. Of course, the pathologies tell us that this is not always the case. Here, we only conjecture that at least in the pathology-free region this intuition holds.\bigskip

Before we conclude this section, let us highlight one more result that further supports the physical relevance of this new approach:

\begin{theorem}
    The map $R:\mathcal{V}_\phys\to\mathcal{V}_\phys$ satisfies\footnote{Note that this limit is well-defined as $R_{n-1}R_n$ itself is a well-defined finite volume RGT. Therefore, the existence of this limit follows from Theorem \ref{Thm1}.}
    \begin{equation}\label{Eq4}
        (R^kH)_m=\lim_{n\to\infty}\omega_{(n-k)m}R_{n-k+1}R_{n-k+2}\cdots R_nH_n.
    \end{equation}
\end{theorem}

\begin{proof}
    We want to prove this result using induction. Let us start by considering the case for $k=2$. The left hand side of Equation (\ref{Eq4}) is given by
    \begin{equation}
        (R^2H)_m=\lim_{k\to\infty}\omega_{k-1m}R_k(RH)_k=\lim_{k\to\infty}\omega_{k-1m}R_k(\lim_{n\to\infty}\omega_{n-1k}R_nH_k).
    \end{equation}
    Since $R_k$ is continuous, we can exchange it with the limit giving us
    \begin{equation}
        (RH^2)_m=\lim_{k\to\infty}\lim_{n\to\infty}\omega_{k-1m}R_k(\omega_{n-1k}R_nH_n).
    \end{equation}
    Comparing this with the right hand side, we are interested in the limit
    \begin{equation}
        \lim_{k\to\infty}\lim_{n\to\infty}\left\|\omega_{k-1m}R_k(\omega_{n-1k}R_nH_n)-\omega_{k-1m}R_k(R_{k+1}H_{k+1})\right\|_\infty.
    \end{equation}
    Applying Lemma \ref{Lem3}, we find that
    \begin{align}
        \left\|\omega_{k-1m}R_k(\omega_{n-1k}R_nH_n)-\omega_{k-1m}R_k(R_{k+1}H_{k+1})\right\|_\infty&\leq\frac{1}{b^{k-1-m}}\|R_k(\omega_{n-1k}R_nH_n)-R_k(R_{k+1}H_{k+1})\|_\infty\nonumber\\
        &\leq\frac{1}{b^{k-1-m}}\|\omega_{n-1k}R_nH_n-R_{k+1}H_{k+1}\|_\infty.
    \end{align}
    In order to bound the remaining expression we can use Lemma \ref{Lem4} together with the fact that $H\in\mathcal{V}_\phys$ giving us
    \begin{equation}
        \frac{1}{b^{k-1-m}}\|\omega_{n-1k}R_nH_n-R_{k+1}H_{k+1}\|_\infty\leq\frac{1}{b^{k-1-m}}\sum_{l=k+2}^n\frac{\|\Delta H_k\|_\infty}{b^{l-1-k}}\leq\frac{Cb^{k+1}}{b^{k-1-m}}\sum_{l=k+2}^nb^{(\alpha-1)l}.
    \end{equation}
    Carrying out the limit $n\to\infty$, we obtain
    \begin{equation}
        \lim_{n\to\infty}\left\|\omega_{k-1m}R_k(\omega_{n-1k}R_nH_n)-\omega_{k-1m}R_k(R_{k+1}H_{k+1})\right\|_\infty\leq Cb^{m+2}\frac{b^{(\alpha-1)(k+2)}}{1-b^{\alpha-1}}.
    \end{equation}
    In the limit $k\to\infty$ this vanishes due to $(\alpha-1)<0$ which completes the proof for $k=2$. Now, let us assume the statement is true for all $k\leq N$ for some $N\geq 2$. We then must prove that it also holds for $k=N+1$. First, we note that
    \begin{equation}
        R^{N+1}H=R^N(RH).
    \end{equation}
    By the assumption of the induction, we can then write this as
    \begin{align}
        (R^N(RH))_m=\lim_{k\to\infty}\omega_{(k-N)m}R_{k-N+1}\cdots R_{k}(RH)_k\nonumber\\
        =\lim_{k\to\infty}\omega_{(k-N)m}R_{k-N+1}\cdots R_{k}(\lim_{n\to\infty}\omega_{n-1k}R_nH_n)\\
        =\lim_{k\to\infty}\lim_{n\to\infty}\omega_{(k-N)m}R_{k-N+1}\cdots R_{k}\omega_{n-1k}R_nH_n.\nonumber
    \end{align}
    As a direct consequence of Lemma \ref{Lem3} we find that
    \begin{equation}
        \|R_{k-N+1}\cdots R_{k}H_k-R_{k-N+1}\cdots R_{k}G_k\|_\infty\leq\|H_k-G_k\|_\infty.
    \end{equation}
    Thus, the proof that
    \begin{equation}
        \lim_{k\to\infty}\lim_{n\to\infty}\omega_{(k-N)m}R_{k-N+1}\cdots R_{k}\omega_{n-1k}R_nH_n=\lim_{k\to\infty}\omega_{(k-N)m}R_{k-N+1}\cdots R_{k}R_{k+1}H_{k+1}
    \end{equation}
    remains essentially the same as in the case for $k=2$. The only difference is that we get an additional factor of $b^N$ in the final estimate, which, however, does not alter the fact that the difference will vanish in the limit $k\to\infty$. This completes the induction.
\end{proof}

This theorem establishes that the regularization of the RGT using the projections $\omega_{nm}$ is compatible with the semi-group structure of the infinite volume RG-kernels. In other words; It does not matter whether we first build the kernel of two consecutive RG-steps and then build the map $R:\mathcal{V}_\phys\to\mathcal{V}_\phys$ or if we instead build $R$ using the kernel of a single step and then consider the map $R^2$. We can also think of it as
\begin{equation}
    R_{T^2}=(R_T)^2.
\end{equation}
It seems plausible that with more effort, we could establish the more general result
\begin{equation}
    R_{T\tilde T}=R_{T}R_{\tilde T}.
\end{equation}
Although, one must be very careful when mixing different block geometries as this framework so far does not allow for that.

\section{Where are we standing?}\label{Sec3}

With the help of the construction of $\mathcal{V}_\infty$ using graph coverings, we managed to obtain a version of the RGT $R:\mathcal{V}_\phys\to \mathcal{V}_\phys$ that is well-defined and continuous on a space of Hamilton functions. This gives us a version of the RGT that is much more in line with Wilson's vision than traditional attempts on the infinite lattice which are plagued by pathologies. Further, we have established a canonical correspondence between conventional infinite volume interactions like the nearest-neighbour Ising model and points in $\mathcal{V}_\phys$ that allows us to start investigating the flow of $R$ in physically relevant models. The next logical steps from here on out are i) establishing Conjecture \ref{Conj1} as a rigorous result and then ii) start to investigate the properties of the flow of $R$ induced in $\mathcal{V}_\phys$. In particular, it is desirable to establish the following properties:
\begin{enumerate}
    \item[a)] The $T=\infty$ fixed point $H=0$ is stable under $R$ and the corresponding basin of attraction contains $H_\beta^\mathrm{Ising}$ for all $\beta<\beta_c$.
    \item[b)] There exists a region in $\mathcal{V}_\phys$ that escapes towards infinity (the $T=0$ fixed point lies at infinity) and contains $H_\beta^\mathrm{Ising}$ for all $\beta>\beta_c$.
    \item[c)] There exists at least one non-trivial fixed point that attracts $H_{\beta_c}^\mathrm{Ising}$.
\end{enumerate}
It should be noted that when we talk about the RG-flow we are interested in the action of $R$ onto the coset $\mathcal{V}_\phys/\mathds{R}$ (see the discussion below). Beyond these properties related to the dynamics of $R$ one is, of course, also interested in more general concepts of Wilson's RG like the extraction of critical exponents and whether they match the known results in cases like the $2\mathrm{D}$ Ising model. In any case, this framework opens up many new pathways to explore towards a fully rigorous version of Wilson's renormalization group. 

\bigskip

Before we conclude this paper, let us dive into some preliminary ideas towards establishing a). Let us denote the quotient space $\hat {\mathcal{V}}:=\mathcal{V}_\phys/\mathds{R}$. $\mathds{R}$ acts onto $\mathcal{V}_\infty$ via
\begin{equation}
    (H+C)_n:=H_n+b^nC,\quad C\in\mathds{R}.
\end{equation}
$\hat{\mathcal{V}}$ is well-defined since $H\in \mathcal{V}_\phys\implies H+C\in \mathcal{V}_\phys,\,\,C\in\mathds{R}$. We can convince ourselves of this fact by computing
\begin{equation}\label{Eq2}
    H_n(\sigma)=\sum_{i=1}^bH_{n-1}(\sigma_{Q_i})+\Delta H_n(\sigma)\implies H_n(\sigma)+C=\sum_{i=1}^b\left[H_{n-1}(\sigma_{Q_i})+\frac{C}{b}\right]+\Delta H_n(\sigma).
\end{equation}
Therefore, adding a constant to the Hamilton function leaves $\Delta H_n$ invariant and hence also its scaling. Further, the RGT $R:\mathcal{V}_\phys\to \mathcal{V}_\phys$ extends to $\hat{\mathcal{V}}$ simply due to the fact that $R(H+C)=R(H)+bC$. In order to speak about convergence of $R$, we should also think about how to extend the norm on $\mathcal{V}_\phys$ to a norm on $\hat{\mathcal{V}}$. We define
\begin{equation}
    \|[\hat H]\|\hat{\,}:=\inf_{H\in[\hat H]}\|H\|.
\end{equation}

\begin{lemma}\label{Lem2}
    There exists a unique element $H\in[\hat H]$ such that $\|H\|=\|[\hat H]\|\hat{\,}$.
\end{lemma}

\begin{proof}
    Let us start with some arbitrary $H'\in[\hat H]$. Similarly to how we proved that the norm $\|\cdot\|$ on $\mathcal{V}_\phys$ is well-defined, we can prove that the extremal points 
    \begin{equation}
        T(H'):=\lim_{n\to\infty}\frac{\max H'_n}{|\Lambda_n|}\qquad;\qquad B(H'):=\lim_{n\to\infty}\frac{\min H'_n}{|\Lambda_n|}
    \end{equation}
    both exist. By construction it is clear that $\|H'\|=\max\{|T(H')|,|B(H')|\}$. Further, we find that $T(H'+C)=T(H')+C$ and likewise for $B(H')$. Thus, we conclude that
    \begin{equation}
        \|[\hat H]\|\hat{\,}=\frac{1}{2}(T(H')-B(H'))=\left\|H'-\frac{1}{2}(T(H')+B(H')\right\|.
    \end{equation}
    Therefore, the Hamiltonian $H=H'-\frac{1}{2}(T(H')+B(H')\in\|[\hat H]\|\hat{\,}$ satisfies the statement. On top of that, we note that 
    \begin{equation}
        T(H)=-B(H)=\|[\hat H]\|\hat{\,}.
    \end{equation}
    Therefore, by adding a constant to $H$, we will either shift $T(H+C)$ upwards if the constant is positive or $B(H+C)$ downwards if it is negative. In either case we find that $\|H+C\|>\|H\|$ such that $H$ is the unique element in $[\hat H]$ with the property $\|H\|=\|[\hat H]\|\hat{\,}$.
\end{proof}

\begin{corollary}
    The extension $R:\hat{\mathcal{V}}\to \hat{\mathcal{V}}$ is continuous regarding $\|\cdot\|\hat{\,}$.
\end{corollary}

\begin{proof}
    Let $[\hat H],[\hat G]\in\hat{\mathcal{V}}$. Further, let $H,G\in \mathcal{V}_\phys$ be the unique representatives from the previous lemma. Then we find that $\|G-H\|=\|[\hat G]-[\hat H]\|\hat{\,}$. Further, in the proof of Theorem \ref{Thm3} we had shown that $R:\mathcal{V}_\phys\to \mathcal{V}_\phys$ is Lipschitz continuous with constant $b$ such that $\|R(G)-R(H)\|\leq b\|[\hat G]-[\hat H]\|\hat{\,}$. At last, we note that under the projection $\mathcal{V}_\phys\to \hat{\mathcal{V}}$ we find that $\|[\hat L]\|\hat{\,}\leq\|L\|$ and hence
    \begin{equation}
        \|[\hat{R(G)}]-[\hat{R(H)}]\|\hat{\,}\leq b\|[\hat G]-[\hat H]\|\hat{\,}.
    \end{equation}
\end{proof}

Now let $[\hat H]\in\hat{\mathcal{V}}$ such that $\|[\hat H]\|\hat{\,}\ll1$. Further, let $H\in[\hat H]$ be some representative. We shall set $h_n:=b^{-n}H_n$. The first step is to compute
\begin{equation}
    \omega_{n-1m}R_nH_n=\frac{-1}{b^{n-1-m}}\ln\left(\sum_sT_n(\pi_{n-1m}^*\sigma|s)e^{-b^nh_n(s)}\right).
\end{equation}
Before we can do that, however, we must establish the following result:

\begin{lemma}
    We find
    \begin{equation}
        \sum_sT_n(\sigma|s)=2^{|\Lambda_n|-|\Lambda_{n-1}|}=:\mathcal{N}_n\,.
    \end{equation}
    In particular, $\sum_sT_n(\sigma|s)$ is independent of $\sigma$.
\end{lemma}

\begin{proof}
    Let us first prove that $\sum_sT_n(\sigma|s)$ is in fact independent of $\sigma$ and then use this fact in order to compute its value. Here we may use that a function $F:\Omega_{\Lambda_{n-1}}\to\mathds{R}$ is constant if and only if it is invariant under flipping an arbitrary spin. Thus, let us consider some $x\in\Lambda_{n-1}$ and $\sigma\in\Omega_{\Lambda_{n-1}}$. We denote the configuration $\tilde\sigma$ defined by $\tilde\sigma_x=-\sigma_x$ while $\tilde\sigma_y=\sigma_y$ for all $y\in\Lambda_{n-1}\backslash\{x\}$. We then define the automorphism
    \begin{equation}
        f:\Omega_{\Lambda_n}\to\Omega_{\Lambda_n},\qquad f(s)_z:=\begin{cases}
            s_z&z\notin\phi_n^{-1}(x)\\
            -s_z&z\in\phi_n^{-1}(x)
        \end{cases}.
    \end{equation}
    By $\mathds{Z}_2$-invariance of the single block kernel $t(\sigma_x|s_{\phi_n^{-1}(x)})$ we find that
    \begin{equation}
        T_n(\tilde\sigma|f(s))=T_n(\sigma|s).
    \end{equation}
    We conclude that
    \begin{equation}
        \sum_sT_n(\sigma|s)=\sum_sT_n(\tilde\sigma|f(s))=\sum_sT_n(\tilde\sigma|s)
    \end{equation}
    where we have used that summing over $s$ is the same as summing over $f(s)$. Now, in order to compute $\mathcal{N}_n$ we may consider
    \begin{equation}
        2^{|\Lambda_n|}=\sum_s\sum_\sigma T_n(\sigma|s)=\sum_\sigma\sum_sT_n(\sigma|s)=2^{|\Lambda_{n-1}|}\mathcal{N}_n.
    \end{equation}
\end{proof}

Let us now use this result in
\begin{align}
    \ln\left(\sum_sT_n(\pi_{n-1m}^*\sigma|s)e^{-b^nh_n(s)}\right)&=\ln\left(\mathcal{N}_n\sum_s\frac{T_n(\pi_{n-1m}^*\sigma|s)}{\mathcal{N}_n}e^{-b^nh_n(s)}\right)\nonumber\\
    &=\ln(\mathcal{N}_n)+\ln\left(\sum_s\frac{T_n(\pi_{n-1m}^*\sigma|s)}{\mathcal{N}_n}e^{-b^nh_n(s)}\right).
\end{align}
As we are working with equivalence classes in $\hat{\mathcal{V}}$, we may throw away the constant contribution coming from the $\ln(\mathcal{N}_n)$ term. Therefore, we consider
\begin{equation}
    \frac{-1}{b^{n-1-m}}\ln\left(\sum_sP_{\pi_{n-1m}^*\sigma}(s)e^{-b^nh_n(s)}\right),\qquad P_{\pi_{n-1m}^*\sigma}(s):=\frac{T_n(\pi_{n-1m}^*\sigma|s)}{\mathcal{N}_n},
\end{equation}
where $P_{\pi_{n-1m}^*\sigma}(s)$ is a probability measure on $\Omega_{\Lambda_n}$. We recognize the $\ln$ term as the cumulant generating function for the random variable $h_n$ evaluated at $t=-b^n$ regarding the probability measure $P_{\pi_{n-1m}^*\sigma}(s)$. This allows us to expand (keep in mind that we removed the constant offset)
\begin{equation}
    \omega_{n-1m}R_nH_n=b^{m+1}\sum_{k=1}^\infty\frac{(-b^n)^{k-1}}{k!}\kappa_k
\end{equation}
where the $\kappa_k$ are the cumulants of $h_n$ regarding $P_{\pi_{n-1m}^*\sigma}(s)$. The $k$th cumulant decays like $\|h_n\|_\infty^k$. By choosing the initial $H$ to be the unique element in $[\hat H]$ satisfying the property from Lemma \ref{Lem2}, we achieve that $\|h_n\|_\infty=\|[\hat H]\|\hat{\,}$. Therefore, the higher order cumulants are small in $\|[\hat H]\|\hat{\,}$. However, there is another issue arising, namely, the scaling with $b^n$ which becomes large very quickly. Thus, the next step would be to use properties from the kernels $T_n$ in order to prove that the cumulants suppress this $b^n$ scaling and the higher order terms really can be treated as corrections that are small $\|[\hat H]\|$. Then, we could start to investigate whether the linear term truly is a contraction or not. For now, this remains a project for the future.

\section*{Acknowledgements}

The author would like to thank his supervisor Uwe-Jens Wiese for initiating this project as part of his master thesis and continuous support ever since. Further gratitude is due to Urs Wenger, Sebastian Baader, Matthias Blau and Nico Scheidegger for fruitful discussions. A special thank you also to Aernout van Enter for reading a previous version of this manuscript and providing helpful comments.

\printbibliography

@book{friedli_velenik_2017,
place={Cambridge},
title={Statistical Mechanics of Lattice Systems: A Concrete Mathematical Introduction},
DOI={10.1017/9781316882603},
ISBN={978-1-107-18482-4},
publisher={Cambridge University Press},
author={Friedli, Sacha and Velenik, Yvan},
year={2017}
}

@article{van_Enter_1993,
   title={Regularity properties and pathologies of position-space renormalization-group transformations: Scope and limitations of Gibbsian theory},
   volume={72},
   ISSN={1572-9613},
   url={http://dx.doi.org/10.1007/BF01048183},
   DOI={10.1007/bf01048183},
   number={5–6},
   journal={Journal of Statistical Physics},
   publisher={Springer Science and Business Media LLC},
   author={van Enter, Aernout C. D. and Fernández, Roberto and Sokal, Alan D.},
   year={1993},
   month={sep}, pages={879–1167} }

@article{Griffiths1979,
  author       = {Griffiths, Robert B. and Pearce, Paul A.},
  title        = {Mathematical properties of position-space renormalization-group transformations},
  journal      = {Journal of Statistical Physics},
  year         = {1979},
  volume       = {20},
  number       = {5},
  pages        = {499--545},
  doi          = {10.1007/BF01012897},
  url          = {https://doi.org/10.1007/BF01012897},
  issn         = {1572-9613}
}

@incollection{Israel1981,
  author       = {Israel, R. B.},
  title        = {Banach algebras and Kadanoff transformations},
  booktitle    = {Random Fields (Esztergom, 1979), Vol. II},
  editor       = {Fritz, J. and Lebowitz, J. L. and Szász, D.},
  pages        = {593--608},
  publisher    = {North-Holland},
  location     = {Amsterdam},
  year         = {1981}
}

@article{doi:10.1137/1113026,
author = {Dobruschin, P. L.},
title = {The Description of a Random Field by Means of Conditional Probabilities and Conditions of Its Regularity},
journal = {Theory of Probability \& Its Applications},
volume = {13},
number = {2},
pages = {197-224},
year = {1968},
doi = {10.1137/1113026},

URL = { 
    
        https://doi.org/10.1137/1113026
    },
eprint = { 
    
        https://doi.org/10.1137/1113026
    }
}

@article{Lanford1969,
  author       = {Lanford, O. E. and Ruelle, D.},
  title        = {Observables at infinity and states with short range correlations in statistical mechanics},
  journaltitle = {Communications in Mathematical Physics},
  date         = {1969-09-01},
  volume       = {13},
  number       = {3},
  pages        = {194--215},
  issn         = {1432-0916},
  doi          = {10.1007/BF01645487},
  url          = {https://doi.org/10.1007/BF01645487}
}

@article{Kadanoff,
  title = {Scaling laws for Ising models near ${T}_{c}$},
  author = {Kadanoff, Leo P.},
  journal = {Physics Physique Fizika},
  volume = {2},
  issue = {6},
  pages = {263--272},
  numpages = {10},
  year = {1966},
  month = {Jun},
  publisher = {American Physical Society},
  doi = {10.1103/PhysicsPhysiqueFizika.2.263},
  url = {https://link.aps.org/doi/10.1103/PhysicsPhysiqueFizika.2.263}
}

@article{WilsonI,
  title = {Renormalization Group and Critical Phenomena. I. Renormalization Group and the Kadanoff Scaling Picture},
  author = {Wilson, Kenneth G.},
  journal = {Phys. Rev. B},
  volume = {4},
  issue = {9},
  pages = {3174--3183},
  numpages = {0},
  year = {1971},
  month = {Nov},
  publisher = {American Physical Society},
  doi = {10.1103/PhysRevB.4.3174},
  url = {https://link.aps.org/doi/10.1103/PhysRevB.4.3174}
}

@article{WilsonII,
  title = {Renormalization Group and Critical Phenomena. II. Phase-Space Cell Analysis of Critical Behavior},
  author = {Wilson, Kenneth G.},
  journal = {Phys. Rev. B},
  volume = {4},
  issue = {9},
  pages = {3184--3205},
  numpages = {0},
  year = {1971},
  month = {Nov},
  publisher = {American Physical Society},
  doi = {10.1103/PhysRevB.4.3184},
  url = {https://link.aps.org/doi/10.1103/PhysRevB.4.3184}
}

@article{Yin,
    author = {Yin, Mey},
    title = {Spectral Properties of the Renormalization Group at Infinite Temperature},
    journal = {Communications in Mathematical Physics},
    volume = {304},
    pages = {175 - 186},
    date = {2011-05-01},
    doi = {10.1007/s00220-011-1201-5},
    url = {https://doi.org/10.1007/s00220-011-1201-5},
}

@article{Slava,
    author = {Kennedy, Tom and Rychkov, Slava},
    title = {Tensor RG Approach to High-Temperature Fixed Point},
    journal = {Journal of Statistical Physics},
    date = {2022-05-02},
    volume = {187},
    doi = {10.1007/s10955-022-02924-4},
    url = {https://doi.org/10.1007/s10955-022-02924-4},
    
}

\end{document}